\documentclass[letterpaper, 10 pt, conference]{ieeeconf}  % Comment this line out if you need a4paper

\IEEEoverridecommandlockouts                              % This command is only needed if 
\usepackage{amsmath}
\usepackage{centernot}
\usepackage{graphicx}
\usepackage{float}
\usepackage{color}
\usepackage{empheq}
\usepackage{amssymb}
\usepackage{cases} 
\usepackage[ruled,vlined]{algorithm2e}

\newtheorem{thm}{Theorem}
\newtheorem{defn}{Definition}

\newtheorem{lem}{Lemma}

\newtheorem{asm}{Assumption}

\newtheorem{remark}{Remark}

\newcommand{\obs}{\mathit{obs}}
\newcommand{\ini}{\text{in}}
\newcommand{\XXX}{\mathcal{X}}
\newcommand{\GGG}{\mathcal{G}}
\newcommand{\SSS}{\mathcal{S}}

\title{\LARGE \bf
Local Opacity Verification for Distributed Discrete Event Systems
}

\author{Sasinee Pruekprasert$^{1}$ and Kai Cai$^{2}$% <-this % stops a space
\thanks{S. P. is supported by ERATO HASUO
			Metamathematics for Systems Design Project No.~JPMJER1603, JST, and Grant-in-aid No. 21K14191, JSPS.	K. C. is supported by Grant-in-aid No. 21H04875, JSPS.
			}% <-this % stops a space
\thanks{$^{1}$Sasinee Pruekprasert is with the National Institute of Informatics, Hitotsubashi 2-1-2, Tokyo 101-8430, Japan.
        {\tt\small sasinee@nii.ac.jp}}%
\thanks{$^{2}$Kai Cai is with Department of Electrical and Information Engineering, Osaka
City University, Osaka 558-8585, Japan.
        {\tt\small kai.cai@eng.osaka-cu.ac.jp}}%
}

\begin{document}
{\onecolumn\large 
\noindent\textcopyright\ 2021 IEEE. Personal use of this material is permitted.  Permission from IEEE must be obtained for all other uses, in any current or future media, including reprinting/republishing this material for advertising or promotional purposes, creating new collective works, for resale or redistribution to servers or lists, or reuse of any copyrighted component of this work in other works.}
\newpage
\twocolumn

\maketitle
\thispagestyle{empty}
\pagestyle{empty}

%%%%%%%%%%%%%%%%%%%%%%%%%%%%%%%%%%%%%%%%%%%%%%%%%%%%%%%%%%%%%%%%%%%%%%%%%%%%%%%%
\begin{abstract}
This paper studies current-state opacity and 
initial-state opacity verification of distributed discrete
event systems.
The distributed system's global model is the parallel
composition of multiple local systems: each of which
represents a component.
We propose sufficient conditions for verifying opacity of the global system model based only on the opacity of the local systems. We also present efficient approaches for the opacity
verification problem that only rely on the intruder's
observer automata of the local DESs.
\end{abstract}

%%%%%%%%%%%%%%%%%%%%%%%%%%%%%%%%%%%%%%%%%%%%%%%%%%%%%%%%%%%%%%%%%%%%%%%%%%%%%%%%
\section{Introduction}

Security is a crucial issue in many applications, especially for %online services and network communication
distributed systems with multiple components that communicate across a network. As a result, methodologies to protect data privacy from malicious intruders are needed.
In this work, we consider the concept of \emph{system opacity}: a property that indicates whether or not a given ``secret'' about the system is detectable by the intruder based on the observed system's behaviors.
Opacity was proposed for analyzing security protocols in \cite{mazare2004using}. 
This concept was introduced to the discrete
event systems (DES) community in \cite{bryans2005modelling} for petri-nets, and in \cite{bryans2008opacity} for transition systems, and has been an active research topic in the DES community in recent years. Several notions of opacity have been proposed and studied in the literature \cite{wu2013comparative, saboori2013verification, falcone2015enforcement}. 
Overviews, surveys on commonly used techniques and existing tools, and historical remarks on the opacity of DES are presented in \cite{jacob2016overview, guo2020overview, lafortune2018history}.
%This paper aims to review the most commonly used techniques of opacity validation for deterministic models and opacity quantification for probabilistic ones. Available complexity results are also provided. Finally, we review existing tools for opacity validation and current applications
%\cite{jacob2016overview, guo2020overview, lafortune2018history}.
%\cite{jacob2016overview, wu2013comparative, falcone2015enforcement, tong2018decentralized, guo2020overview}.
%Decentralized opacity enforcement using supervisory control was also studied in \cite{tong2018decentralized}.
%Our work is different from \cite{tong2018decentralized}, which consider decentralized architecture with multiple intruders that observe the system through different observation maps.
%Instead, we consider one intruder that aims to attack the global DES using one observation map.

This work studies the  opacity verification of distributed DESs, which are systems with modular structure as illustrated in Fig.~\ref{fig: architecture}.  
Opacity verification for modular systems is known to be decidable but computationally expensive: its complexity has shown to be in EXPSPACE-complete for general cases, and PSPACE-complete  if all events shared by any local DESs are observable \cite{masopust2019complexity}.
The monolithic approach to verify the opacity of modular systems is to construct data structures that estimate the intruder's global DES information based on observed event sequences.
These data structures can be large,  especially for
distributed DESs with several local components. 
To mitigate this problem, many previous studies (e.g., \cite{mohajerani2019compositional, noori2018compositional, mohajerani2019transforming})
investigated compositional opacity verification approaches. 
The approaches construct abstract structures of local DESs.
Then, the opacity verification is performed by considering synchronizations between those abstract structures in a sound and efficient manner.

In this work, we propose a different approach to verify the opacity of distributed DESs, which we call \emph{local opacity verification}.
Our technique verifies the opacity of the global system of the distributed DESs by estimating the intruder's information %to 
%verify the opacity of the  distributed DESs 
\emph{based only on local DESs}.
We propose 
sufficient conditions and their corresponding efficient verification approaches
 for the global DES's opacity using local DESs' observer automata. 
 %\emph{based only on  local DESs}.
 %without constructing an observer automaton of the global DES.
  We focus on current-state opacity (CSO) and initial-state opacity (ISO) verification. 
Our results for these two types of opacity will establish a foundation for studying other opacity notions in the future.
%By assuming that the events shared between local DESs are observed by the intruder,
%we proposed sufficient conditions for the opacity (CSO and ISO) of the global DES, by considering only the opacity of local DESs. Using these sufficient conditions, we introduced efficient methodologies
%to verify the opacity of the global DES without constructing the intruder's observer of the global DES.

%\begin{table}[t]
%\caption{Comparison of distributed/modular DES architectures}
%\label{table compare}
%\begin{center}\scriptsize
%\begin{tabular}{|c||c|c|c|c|}
%\hline
%Work & Opacity & Intruder & Observation&Secret-state set\\
% &   & & map& $(\SSS[i]$ is a local secret set)\\
%\hline
%\cite{tong2018decentralized} & CSO & multiple & multiple & $\SSS \subseteq \bigtimes_i \SSS[i]$\\
%\hline 
%\cite{zinck2020enforcing} & CSO & one & global & $\bigcup_i (\SSS[i] \times \bigtimes_i X_i)$\\
%\hline 
%\cite{saboori2010reduced} & ISO & one & global & $\SSS =\bigtimes_i \SSS[i]$\\
%\hline
%\cite{tong2019current} & CSO & one & local & $\SSS \subseteq \bigtimes_i \SSS[i]$\\
%\hline
%Ours & CSO, ISO & one & global & $\SSS \subseteq \bigtimes_i \SSS[i]$\\
%\hline 
%\end{tabular}
%\end{center}
%\vspace{-0.5cm}
%\end{table} 
%Table \ref{table compare} compares our distributed DES setting to related works on the verification and enforcement of DESs with modular/decentralized architecture.

Reduced complexity opacity verification techniques for modular DESs by considering its local components was proposed for ISO in \cite{saboori2010reduced},
and for CSO in \cite{tong2019current} and \cite{zinck2020enforcing}.
However, our setting is different from the previous studies.
The objective of \cite{saboori2010reduced} and \cite{zinck2020enforcing} is to ensure that
 local secret states of all local DESs are protected confidentially.
 In both studies, 
 the secret of the whole system is revealed if a secret state of any local DES is revealed. 
Instead, in our work, we consider the secret of the global system defined as a subset of the global system's states, as depicted in Fig.~\ref{fig: architecture}.
Our work is also different from \cite{tong2019current}, which considers one observation map for each local DES.
Using these maps, the intruder observes the global DES through the event sequences of all local DESs. 
In our setting, the intruder has only one observation map and observes the global DES's event sequences  directly.
These differences make the algorithms developed in the previous studies unsuitable to solve our opacity verification problem.
%one cannot straightforwardly generalize our setting to those in the previous works.
 
% in \cite{tong2019current}, an observation map $\pi_{\obs,i}: \Sigma_\GGG^* \to \Sigma_{\obs,i}^*$ is defined for each local DES $G_i$, and then the notion of CSO is defined with respect to all local observation maps.
%Therefore, in their setting, the intruder does not observe the behavior of the global DES, but observe the behaviors of all local DESs. 
%In our work, as in Definition~\ref{defn CSO}, we define the observation map as a function $\pi_\obs: \Sigma_\GGG^* \to \Sigma_\obs^*$, which means that the intruder observes the global DES $\GGG$ directly.
%This difference is significant: one cannot generalize our setting to \cite{tong2019current}, and vice versa, by simply use $\pi_{\obs,i} = \pi_\obs \circ \pi_i$. For example, let $a \in \Sigma_1 \cap \Sigma_\obs$ and $b \in \Sigma_2 \cap \Sigma_\obs$, then  
%$\pi_\obs \circ \pi_i(ab) = \pi_\obs \circ \pi_i(ba)$ for $i \in \{1,2\}$, but $\pi_\obs(ab) \neq \pi_\obs(ba)$. 

The rest of this paper is organized as follows.
In Section~\ref{section DDES}, we propose a distributed architecture of a global DES with several local DESs. %as its components. 
In Section~\ref{section problem},
we  %present the notions of CSO and ISO, and 
introduce the concept of local opacity verification.
We present sufficient conditions for the global DES's opacity in Section~\ref{section local},  %based only on the local DESs. 
%Using these conditions %in Section~\ref{section local}, 
and propose the corresponding approaches to verify the global DES using the observer automata of local DESs in Section~\ref{section algo}.
Finally, Section~\ref{section conclusion} presents the conclusion.

\begin{figure}[t]
  \centering
    \includegraphics[width=1\linewidth]{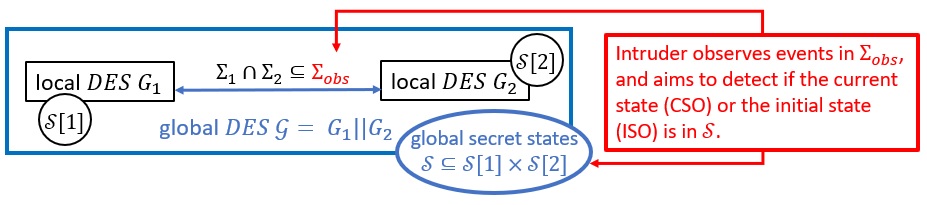} 
  \caption{An overview of our distributed architecture with two local DESs. The global DES $\GGG$ is the parallel composition of local DESs. The secret state set $\SSS$ is a subset of states of $\GGG$.
 The intruder observes  the events shared between the local DESs and some local events, then
%We assume that the events shared by the local DESs are observable by the intruder ($\Sigma_1 \cap \Sigma_2 \subseteq \Sigma_\obs$). 
 aims to detect if the current state (for the CSO) or the initial state (for the ISO) of $\GGG$ is in $\SSS$ (see Section \ref{section problem}). Note that this architecture can be extended to $n$ local DESs.
  }
  \label{fig: architecture} 
\end{figure}

\section{Distributed Discrete Event Systems}\label{section DDES}

We study a distributed discrete event system (DES), which we call the global DES, consisting of $n$ local DESs.  Fig.~\ref{fig: architecture} depicts an overview of the architecture with two local DESs.  

\subsection{Local and Global Systems}
For each $i \in \{1,\ldots,n\}$, we model the local DES $G_i$ as 
\begin{equation*}
G_i=(X_i,\Sigma_i,\delta_i,X_{\ini,i}),
\end{equation*}
where 
$X_i$ is the set of states, $\Sigma_i$ is the set of events,
$\delta_i: X_i \times \Sigma_i \rightharpoonup X_i$ is a partial transition function,
and $X_{\ini,i} \subseteq X_i$ is the set of initial states.
We use the notation $\delta_i(x,\sigma)!$ for ``$\delta_i(x,\sigma)$ is defined''. 
%and $\neg\delta_i(x,\sigma)!$ for ``$\delta_i(x,\sigma)$ is undefined''.
We also write $\Sigma_{G_i}$ (\emph{resp.} $\delta_{G_i}$ )  for $\Sigma_i$ (\emph{resp.} $\delta_i$) when we specifically refer to it as the event set (\emph{resp.} the transition function) of the DES $G_i$.
 
%The global DES is a distributed system constructed by the parallel composition of all $n$ local DESs.
We consider a distributed system that consists of all $n$ local DESs, %as its components, 
The global DES of the distributed system is constructed by the parallel composition. %of their local DESs.
%Formally, the global DES is the  defined as follows.
\begin{defn}[\cite{cassandras2009}]\label{defn product}
Given two DESs $G_i$ and $G_j$, their parallel composition is the DES
\[G_i \parallel G_j = (X_i \times X_j, \Sigma_{G_i \parallel G_j}, 
		\delta_{G_i \parallel G_j},
		X_{\ini,i} \times X_{\ini,j}),		
		\]
where $\Sigma_{G_i \parallel G_j} = \Sigma_{G_i} \cup \Sigma_{G_j}$ and the transition function $\delta_{G_i \parallel G_j}: X_i \times X_j \times \Sigma_{G_i \parallel G_j} \rightharpoonup  X_i \times X_j$ is defined as follows. 
\begin{subnumcases}{\delta_{G_i \parallel G_j}(x_i, x_j,\sigma) = }
(\delta_{G_i}(x_i,\sigma), \delta_{G_j}(x_j,\sigma))\label{defn product eq1}\\
 	\qquad\text{if } \delta_{G_i}(x_i,\sigma)!\text{ and } \delta_{G_j}(x_j,\sigma)!\nonumber\\
(\delta_{G_i}(x_i,\sigma), x_j)\label{defn product eq2}\\
 	\qquad\text{if } \delta_{G_i}(x_i,\sigma)! \text{ and } \sigma \notin \Sigma_{G_j}\nonumber\\
(x_i, \delta_{G_j}(x_j,\sigma))\label{defn product eq3}\\
	\qquad\text{if } \delta_{G_j}(x_j,\sigma)!\text{ and } \sigma \notin \Sigma_{G_i}\nonumber\\
\text{undefined otherwise}.\label{defn product eq4}
\end{subnumcases}
\end{defn}

Let $G_i \parallel G_j \parallel G_k=G_i \parallel (G_j \parallel G_k)$.
From Definition~\ref{defn product}, the parallel composition of two DESs is also a DES.
Moreover, the composition is associative and commutative up to a reordering of the state components in composed states \cite{cassandras2009}, i.e., $G_i \parallel (G_j \parallel G_k) = (G_i \parallel G_j) \parallel G_k$, and $G_i \parallel G_j$ can be obtained from $G_j \parallel G_i$ by reordering the state components.
In this work, we consider indexed local DESs, so
  we can treat $G_i \parallel G_j$ and $G_j \parallel G_i$ as equivalent.

The global DES is the parallel composition %of all local DESs $G_1,\ldots,G_n$:
\[\GGG = (\XXX, \Sigma_\GGG, \Delta, \XXX_\ini) = G_1 \parallel \ldots \parallel G_n,\]
where $\XXX= X_1 \times \cdots \times X_n$, $\Sigma_\GGG = \Sigma_1 \cup \cdots \cup \Sigma_n$, $\XXX_\ini = X_{\ini, 1} \times \cdots \times X_{\ini, n}$, and 
$\Delta = \delta_{G_1 \parallel G_2 \parallel \ldots \parallel G_n}$.  
%Recall that the parallel composition is associative and is commutative up to a reordering of the state components in composed states, which means
%if $\Delta(x_1, x_2, x_3, \sigma) = (s_1, s_2, s_3)$, then $\Delta(x_2, x_1, x_3, \sigma) = (s_2, s_1, s_3)$.  

\subsection{Extended Transition Functions and Event Sequences}
For any DES $G=(X,\Sigma_G,\delta_G,X_{\ini})$, which can either be a local DES or a composition of local DESs,
we extend its transition function 
$\delta_G$ to $\delta_G^*: X \times \Sigma_G^* \rightharpoonup X$ in the usual way. %(``$*$'' is the Kleene star).
Namely, $\delta_G^*(x,  \varepsilon) = x$, and for all $(\beta, \sigma) \in \Sigma_G^* \times \Sigma_G$,
\begin{align} \label{defn eq delta_star} 
\delta_G^*(x, \beta \sigma) = 
\begin{cases}
\delta_G(\delta_G^*(x, \beta), \sigma) &\text{if } \delta_G^*(x, \beta)!\text{ and}\\
	&~~~~\delta_G(\delta_G^*(x, \beta), \sigma)!\\
\text{undefined}&\text{otherwise.}
\end{cases} 
\end{align} 

%For any state $x \in X$,
%let $L(G, x)$ denote the language (set of event sequences) generated by the DES $G$ starting from $x$, i.e.,
%$L(G, x)= \{ \alpha \in \Sigma^* \mid  \delta_G^*(x,  \alpha)!\}$, and  
% $L(G) = \bigcup_{x \in X_\ini} L(G, x)$.
%Let $L(G) = \bigcup_{ x\in X_\ini} L(G, x)$. 
An event sequence $\alpha$ is generated by $G$ if there exists $x \in X_\ini$ such that $\delta_G^*(x, \alpha)$!.
Let $|\alpha|=k$ denote the length of the event sequences $\alpha = \sigma_1 \ldots \sigma_k \in \Sigma^*$.
%As $\Sigma_G \subseteq \Sigma_\GGG$, 
Let $\pi_G: \Sigma_\GGG^* \to \Sigma_G^*$ be the natural mapping
from event sequences generated by the global DES $\GGG$ to those generated by the DES $G$. 
More precisely, 
$\pi_G(\sigma) = \varepsilon$ if $\sigma = \varepsilon$ or $\sigma \in \Sigma_\GGG \setminus \Sigma_G$, $\pi_G(\sigma) = \sigma$ if $\sigma \in \Sigma_G$, and
$\pi_G(\beta \sigma) = \pi_G(\beta) \pi_G(\sigma)$
for all $(\beta, \sigma) \in \Sigma_\GGG^* \times \Sigma_\GGG$.
For notational convenience, we also use $\pi_i$ for denoting the mapping $\pi_{G_i}$, for $i \in \{1, \ldots, n\}$. %to event sequences of the local DES $G_i$.
Thereby, $\pi_i$ maps each event sequence generated by the global DES $\GGG$ its corresponding sequence generated by the local DES $G_i$. 
% we can use $\pi_i$ to map each event sequence generated by the global DES $\GGG$ to its corresponding sequence generated by the local DES $G_i$.

For a global state $x=(x_1, \ldots, x_n) \in \XXX$, let $x[i]$ denote the local state $x_i \in X_i$.
From Definition~\ref{defn product}, we have the following Lemma.
\begin{lem}\label{lem product1} 
For any $x \in \XXX_{\ini}$ and any $\alpha \in \Sigma^*$,
\begin{equation}\label{lem product1 eq01}
\Delta^*(x, \alpha)! \text{ if and only if } (\delta_i^*(x[i], \pi_i(\alpha))! ,\forall i \in \{1,\ldots n\}).
\end{equation}  
Moreover, if $\Delta^*(x, \alpha)!$,
\begin{equation}\label{lem product1 eq02}
 \Delta^*(x, \alpha) = (\delta_1^*(x[1], \pi_1(\alpha)), \ldots, \delta_n^*(x[n], \pi_n(\alpha)) )
\end{equation}
\end{lem}

\section{Problem formulation}\label{section problem}
\subsection{Notion of Opacity}
In this paper, we consider two notions of opacity: current-state opacity (CSO) and initial-state opacity (ISO), which are two basic types of opacity properties in
the literature \cite{wu2013comparative}. The study of these two types of opacity will lay a foundation for investigation of  more complicated opacity notions.

Let $\Sigma_\obs \subseteq \Sigma_\GGG$ be the set of observable events, and $\pi_\obs: \Sigma_\GGG^* \to \Sigma_\obs^*$ be the observation mapping from each event sequence of $\GGG$ to the  sequence observed by the intruder. Notice that the composite mappings $\pi_i\circ\pi_\obs$ and $
\pi_\obs\circ\pi_i$ are equal. 
%$\pi_i\circ\pi_\obs(\alpha) = \pi_\obs\circ\pi_i(\alpha), \text{ for all } \alpha \in \Sigma_\GGG^* \text{ and all } i \in \{1,\ldots,n\}$. 
Let 
$G=(X,\Sigma_G,\delta_G,X_{\ini})$ be a local DES or a composition of local DESs. 
%We write w.r.t. for ``with respect to''.
\begin{defn}[CSO] \label{defn CSO}
Given a set $S \subseteq X$ of secret states, the DES $G$ is \emph{current-state opaque (CSO)} w.r.t.  $S$  if,
for all $(x, \alpha) \in X_\ini \times \Sigma_G^*$ such that $\delta_G^*(x, \alpha) \in S$, 
there exists $(x', \alpha') \in X_\ini \times \Sigma_G^*$ such that $\delta_G^*(x', \alpha') \in X \setminus S$ and $\pi_{\obs}(\alpha) = \pi_{\obs}(\alpha')$.
\end{defn}

\begin{defn}[ISO] \label{defn ISO}
Given a set $S \subseteq X_\ini$ of secret initial states, 
the DES $G$ is \emph{initial-state opaque (ISO)} w.r.t. $S$ if 
for all $(x, \alpha) \in S \times \Sigma_G^*$ with $\delta_G^*(x, \alpha)!$,
there exists $(x', \alpha') \in (X \setminus S) \times \Sigma_G^*$ with $\delta_G^*(x', \alpha')!$ and $\pi_{\obs}(\alpha) = \pi_{\obs}(\alpha')$. 
\end{defn}

The intuitions of CSO and ISO are as follows. 
CSO (\emph{resp.} ISO) requires that for each event sequence $\alpha$ reaching (\emph{resp.} starting from) a secret state, there exists another sequence $\alpha'$ with the same observable events ($\pi_{\obs}(\alpha) = \pi_{\obs}(\alpha')$) reaching (\emph{resp.} starting from) a non-secret state. 

By Definitions~\ref{defn CSO} and \ref{defn ISO},
we have Lemmas~\ref{lem S1S1} and \ref{lem S1S2}, which state that we can verify the opacity w.r.t. a set $S$ by considering its  subsets $S_1, \ldots, S_m$ such that $S = S_1 \cup \cdots \cup S_m$. 
%This lemma will be used in Section~\ref{section local}.
%We will use this lemma to verify the global DES in Section~\ref{section local}.
\begin{lem}\label{lem S1S1}
$G$ is CSO w.r.t. $S$ if for all $j \in \{1,\ldots,m\}$ and
 all $(x,\alpha) \in X_\ini \times \Sigma_G^*$ such that $\delta_G^*(x, \alpha) \in S_j$,   there exists  $(x', \alpha') \in X_\ini\times \Sigma_G^*$ such that $\delta_G^*(x', \alpha') \in X \setminus S$ and $\pi_{\obs}(\alpha) = \pi_{\obs}(\alpha')$.

%
%Suppose that for all $i \in \{1,\ldots,m\}$ and
% all $(x,\alpha) \in X_\ini \times \Sigma_G^*$ such that $\delta_G^*(x, \alpha) \in S_i$,   there exists  $(x', \alpha') \in X_\ini\times \Sigma_G^*$ such that $\delta_G^*(x', \alpha') \in X \setminus S$ and $\pi_{\obs}(\alpha) = \pi_{\obs}(\alpha')$. Then, $G$ is CSO w.r.t. $S$
\end{lem} 
\begin{lem}\label{lem S1S2}
$G$ is ISO w.r.t. $S$ for all $j \in \{1,\ldots,m\}$ and all
$(x, \alpha) \in S_j \times \Sigma_G^*$ such that $\delta_G^*(x, \alpha)!$,
 there exists $(x', \alpha') \in (X_\ini \setminus S) \times \Sigma_G^*$ such that $\delta_G^*(x', \alpha')!$ and $\pi_{\obs}(\alpha) = \pi_{\obs}(\alpha')$. 
 
%Suppose that for all $i \in \{1,\ldots,m\}$ and all
%$(x, \alpha) \in S_i \times \Sigma_G^*$ such that $\delta_G^*(x, \alpha)!$,
%there exists $(x', \alpha') \in (X_\ini \setminus S) \times \Sigma_G^*$ with $\delta_G^*(x', \alpha')!$ and $\pi_{\obs}(\alpha) = \pi_{\obs}(\alpha')$. 
%Then, $G$ is ISO w.r.t. $S$
\end{lem} 

Lemma~\ref{lem S1S1} (\emph{resp.} Lemma~\ref{lem S1S2}) implies that if for each event sequence $\alpha$ reaching (\emph{resp.} starting from) any secret subset $S_j$, there exists 
another sequence $\alpha'$ with the same observable events (\emph{resp.} starting from) a non-secret state $X \setminus S$, then the DES $G$ is opaque. 
These two lemmas follows from Definitions~\ref{defn CSO} and \ref{defn ISO} and the fact that 
$S = S_1 \cup \cdots \cup S_m $.

%Note that Definitions~\ref{defn CSO} and \ref{defn ISO} present the general definition of opacity in \cite{wu2013comparative}, where the opacity is defined w.r.t. a pair of secret and non-secret sets $(S, N)$.
%However, to be consistent with most of the existing literature \cite{masopust2019complexity,  tong2019current, tong2018decentralized, saboori2010reduced,  zinck2020enforcing,  tong2018current},
%we focus on a more common notion of opacity that considers any state  not included in $S$ to be non-secret.
%\begin{defn}
%Given a set $S \subseteq X$ of secret states,
%  $G$ is opaque (CSO, ISO) w.r.t. $S$ if $G$ is opaque w.r.t. $(S, X\setminus S)$.
%\end{defn}
%We say that 
%(i.e., the any state that is not in $S$ is considered to be non-secret). 
%For the rest of the paper, we will use this common notion of opacity unless stated otherwise.

%\begin{defn}[Common notion of opacity] \label{defn common}
%Given a set $S \subseteq X_\ini$ of secret initial states and a set $N \subseteq X_\ini$ of non-secret initial states, 
%the DES $G$ is \emph{initial-state opaque (ISO)} w.r.t. ($S$, $N$) if 
%for each pair $(x, \alpha) \in S \times \Sigma_G^*$ such that $\delta_G^*(x, \alpha)!$,
%there exists $(x', \alpha') \in N \times \Sigma_G^*$ with $\delta_G^*(x', \beta)!$ and $\pi_{\obs}(\alpha) = \pi_{\obs}(\alpha')$. 
%\end{defn}

\subsection{Opacity Verification Problem}
The goal of this work is to verify whether or not the global DES $\GGG$ is opaque (CSO, ISO) w.r.t. 
a given set $\SSS \subseteq \XXX$ of secret states. 
Note that we allow the set $\SSS$ to be any subset of the set $\XXX$ of global states without any restrictions.
%which is different from the setting of previous studies (e.g., \cite{mohajerani2019transforming, saboori2010reduced, zinck2020enforcing}) that have some restrictions on $\SSS$.
%, and 2) enforce the global DES to be opaque. %We consider the two following problems.
%We define the verification problem as follows.

\begin{defn}[Opacity verification problem]
Given local DESs $G_1, \ldots, G_n$, an observation map $\pi_\obs$, and a secret subset $\SSS \subseteq \XXX$ of global states, 
the opacity verification problem is to 
verify whether the global DES $\GGG = \parallel_{i \in \{1,\ldots,n\}} G_i$ is  opaque (CSO, ISO) w.r.t. the secret set $\SSS$.
\end{defn}

Opacity verification for modular systems is decidable but costly \cite{masopust2019complexity}. %(EXPSPACE-complete in general, and PSPACE-complete %if Assumption~\ref{asm Sig_obs_shared} holds \cite{masopust2019complexity}).
%if all events shared by any local DESs are observable \cite{masopust2019complexity}). 
The monolithic approach %for modular systems opacity verification
 is to construct an observer automaton that estimate the intruder's information of the global DES based on observed event sequences, 
% e.g., the intruder's estimate function in \cite{tong2019current} and the Augmented I-Observer in \cite{tong2018decentralized}.
%These data structures are constructed using the global DES, and 
which can be large.
%, especially when we have many local DESs. 
Therefore, in this work, we consider 
\emph{local opacity verification}: to verify the global DES 
 based only on the observer automata of local DESs.
We propose sufficient conditions for local opacity verification in Sections~\ref{section local} and corresponding efficient approaches in Section~\ref{section algo}.

\section{Local Opacity Verification} \label{section local}
%In this section, we present sufficient conditions for the opacity of the global system %$\GGG$ based on local DESs.
%, based only on the opacity of local DESs $G_i$, $i \in \{1,\ldots,n\}$.

\subsection{Shared Events}
%We first introduce an assumption on shared events, which is necessary for our results.
As depicted in Fig.\ \ref{fig: architecture},
we assume that the intruder can observe events shared between local DESs. 
This assumption is common for DESs with modular structure \cite{saboori2010reduced, contant2006diagnosability}. 
%Then, we show an important property that is holds by this assumption.
%, i.e., we assume that Assumption~\ref{asm Sig_obs_shared} holds. 
\begin{asm}\label{asm Sig_obs_shared}
For all
$\sigma \in \Sigma_\GGG$, we have $\sigma \in \Sigma_{\obs}$ if  there exist
$i,j \in \{1,\ldots,n\}$ such that $i \neq j$ and
$\sigma \in \Sigma_i \cap \Sigma_j$.
\end{asm}

Note that we allow internal events of local DESs to be observable, i.e., we do not require 
$(\Sigma_i \setminus \bigcup_{j \neq i} \Sigma_j) \cap \Sigma_\obs = \emptyset$.

%Assumption~\ref{asm Sig_obs_shared} is a necessary condition for local opacity verification.
%We will discuss this matter in details later on in Remarks \ref{remark CSO Asm1} and \ref{remark ISO Asm1}.
%Under Assumption~\ref{asm Sig_obs_shared},
Then, we present Lemma \ref{lem CSO1}, which will be used for our results in the next sections.
The intuition of this lemma is that
an event sequence $\alpha_i'$ of a local DES $G_i$
can be projected to a sequence of the global DES $\GGG$ (not blocked by the parallel composition)
if there exists at least one sequence $\alpha$ of $\GGG$
 with $\pi_\obs \circ \pi_i(\alpha) = \pi_\obs(\alpha_i')$.

\begin{lem} \label{lem CSO1} 
%Given a set $\SSS \subseteq \XXX$ of global secret states, we assume that Assumption~\ref{asm Sig_obs_shared} holds and
%there exists $(x, \alpha) \in \XXX_\ini \times \Sigma_\GGG^*$ such that $\Delta^*(x, \alpha)!$.
Let $\SSS \subseteq \XXX$ be a set of global secret states.
We assume that Assumption~\ref{asm Sig_obs_shared} holds and there exists $(x, \alpha) \in \XXX_\ini \times \Sigma_\GGG^*$ where $\Delta^*(x, \alpha)!$,
and consider any $i \in \{1,\ldots,n\}$.
For all $(x_i', \alpha_i') \in X_{\ini, i} \times \Sigma_i^*$  satisfying
\begin{equation}\label{lem CSO1 eq1}  
 \delta_i^*(x_i', \alpha_i') ! \text{ and } \pi_\obs(\alpha_i') = \pi_\obs \circ \pi_i(\alpha),
\end{equation}
there exists $ \alpha' \in \Sigma_\GGG^*$ such that $\pi_\obs(\alpha') = \pi_\obs(\alpha)$ and
\begin{align}
\label{lem CSO1 eq2} 
\begin{split}
\Delta^*&(x[1],\ldots,x[i-1],x_i', x[i+1], \ldots, x[n], \alpha')\\
&= (s_1, \ldots, s_{i-1}, \delta_i^*(x_i', \alpha_i'), s_{i+1}, \ldots, s_n),
\end{split}
\end{align}
where $s_j = \delta_j^*(x[j], \pi_j(\alpha))$ for all $j \in \{1,\ldots,n\}$.
\end{lem}
\begin{proof} 
As the parallel composition operation is commutative and associative, 
we assume without loss of generality  that $i = 1$ and 
$\GGG = G_1 \parallel  \GGG_J$, where $\GGG_J = \parallel_{k \in \{2, \ldots, n\}} G_k$. 
Let $x_J = (x[2],\ldots, x[n])$.
%Hence, we can write $x = (x[1], x_J)$, where $x_J = (x[2],\ldots, x[n])$.
We consider  any $(x_1', \alpha_1') \in X_{\ini, 1} \times \Sigma_1^*$ satisfying \eqref{lem CSO1 eq1},
and will prove the lemma by induction on the length of $\alpha$. %(denoted by  $|\alpha|$). 

For the base step, let $\alpha = \varepsilon$ and $\alpha' =  \alpha_1'$.
By \eqref{lem CSO1 eq1},
  we have $\pi_\obs(\alpha') = \pi_\obs \circ \pi_1(\alpha) = \varepsilon =\pi_\obs(\alpha)$.
By  Assumption \ref{asm Sig_obs_shared}, $\alpha_1' \in (\Sigma_1 \setminus \bigcup_{k \in \{1,\ldots,n\}} \Sigma_k)^*$ and $(s_2,\ldots, s_n) = x_J$. By Lemma~\ref{lem product1},
$
\Delta^*(x_1', x_J, \alpha_1') 
 = (\delta_1^*(x_1', \alpha_1'), x_J),
$
and \eqref{lem CSO1 eq2} holds.

%which implies \eqref{lem CSO1 eq2}.

Induction hypothesis (I.H.): if $|\alpha| < k$, there exists $\alpha' \in \Sigma_\GGG^*$ that satisfies $\pi_\obs(\alpha') = \pi_\obs(\alpha)$ and \eqref{lem CSO1 eq2}.

For the inductive step, let $\alpha = \beta \sigma$, where $\sigma \in \Sigma_\GGG$ and $|\beta| < k$. 
%By the induction hypothesis, there exists $ \beta'$
%\begin{equation}\label{lem CSO1 eq3}
%\Delta^*(x_1', x_J, \beta') = (\delta_1^*(x_1', \alpha_1'), \delta_{\GGG_J}^*(x_J,\pi_{\GGG_J}(\beta) )).
%\end{equation} 
%Recall that we assume $\GGG = G_1 \parallel (\GGG_J = \parallel_{k \in \{2 \ldots, n\}} G_k)$ and $\Delta^*(x, \alpha) = s \in S$ by \eqref{lem CSO1 eq1}.
Notice that, since $\Delta^*(x, \alpha)!$, we have $\delta_{\GGG_J}^*(x_J,\pi_{\GGG_J}(\alpha))!$ by Lemma \ref{lem product1}, and $\Delta^*(x, \beta)!$ and $\Delta(\Delta^*(x, \beta), \sigma)!$ by \eqref{defn eq delta_star}. 
%Let us consider any $(x_1', \alpha_1') \in X_{\ini, 1} \times \Sigma_1^*$ that satisfies \eqref{lem CSO1 eq1}.
%To show \eqref{lem CSO1 eq2}, 
We consider the following cases.

\begin{itemize}
\item Case 1: $\sigma \in \Sigma_{\GGG_J} \setminus \Sigma_1$. In this case,  
$\pi_\obs(\alpha_1') = \pi_\obs \circ \pi_1(\alpha) = \pi_\obs \circ \pi_1(\beta)$.
By the I.H., there exists $\beta' \in \Sigma_\GGG^*$ such that $\pi_\obs(\beta' )= \pi_\obs(\beta )$ and
\begin{align}  \label{lem CSO1 eq4}
\Delta^*(x_1', x_J, \beta')
&= (\delta_1^*(x_1', \alpha_1'),   \delta_{\GGG_J}^*(x_J,\pi_{\GGG_J}(\beta))) 
\end{align} 
By setting $\alpha' = \beta' \sigma$, we have $\pi_\obs(\alpha') = \pi_\obs(\beta'\sigma)= \pi_\obs(\beta\sigma) =\pi_\obs(\alpha)$. As $\sigma \notin \Sigma_1$, by \eqref{defn product eq3},
\eqref{defn eq delta_star}, and  \eqref{lem CSO1 eq4}, %and Lemma \ref{lem product1},
\begin{align*}   
\Delta^*(x_1', x_J, \alpha') &= \Delta^*(x_1', x_J, \beta'\sigma)\\
&= (\delta_1^*(x_1', \alpha_1'),   \delta_{\GGG_J}^*(x_J,\pi_{\GGG_J}(\beta \sigma)))\\ 
&= (\delta_1^*(x_1', \alpha_1'),   \delta_{\GGG_J}^*(x_J,\pi_{\GGG_J}(\alpha))). 
\end{align*} 
%which implies \eqref{lem CSO1 eq2}.
Thereby, \eqref{lem CSO1 eq2} holds in this case.

\item Case 2: $\sigma \in \Sigma_1 \cap \Sigma_{\GGG_J}$. 
By Assumption \ref{asm Sig_obs_shared}, %$\pi_\obs \circ \pi_1(\sigma) = \sigma$ and 
$\pi_\obs \circ \pi_1(\alpha) = \pi_\obs \circ \pi_1(\beta) \sigma$.
By \eqref{lem CSO1 eq1},
there exists $\beta_1' \in \Sigma_1^*$ such that $\beta_1' \sigma = \alpha_1'$  and  
$ 
\pi_\obs(\beta_1') \sigma = \pi_\obs(\alpha_1') = \pi_\obs \circ \pi_1(\alpha) = \pi_\obs \circ \pi_1(\beta) \sigma.
$
Thus, we have $\pi_\obs(\beta_1') = \pi_\obs \circ \pi_1(\beta)$. By I.H., there exists $\beta'$ with $\pi_\obs(\beta') = \pi_\obs(\beta) $ and 
\begin{align}\label{lem CSO1 eq2.5} 
\Delta^*(x_1', x_J, \beta') 
&= (\delta_1^*(x_1', \beta_1'),   \delta_{\GGG_J}^*(x_J,\pi_{\GGG_J}(\beta)) ).
\end{align} 
By setting $\alpha' = \beta' \sigma$, we have $\pi_\obs(\alpha') = \pi_\obs(\beta') \sigma = \pi_\obs(\beta) \sigma = \pi_\obs(\alpha) $. Then, by \eqref{defn product eq1},  \eqref{defn eq delta_star}, and \eqref{lem CSO1 eq2.5}, %and Lemma~\ref{lem product1}, 
\begin{align*}  
\Delta^*(x_1', x_J, \alpha') &= \Delta^*(x_1', x_J,\beta' \sigma) \\ 
&= (\delta_1^*(x_1', \beta_1' \sigma),   \delta_{\GGG_J}^*(x_J,\pi_{\GGG_J}(\beta \sigma)) )\\
&= (\delta_1^*(x_1', \alpha_1' ),   \delta_{\GGG_J}^*(x_J,\pi_{\GGG_J}(\alpha)) ).  
\end{align*} 
Thereby, \eqref{lem CSO1 eq2} holds in this case.
%which implies \eqref{lem CSO1 eq2}. 
\item Case 3: $\sigma \in (\Sigma_1 \setminus \Sigma_{\GGG_J}) \cap \Sigma_\obs$.
As $\sigma \in \Sigma_1 \cap \Sigma_\obs$, % we have $\pi_\obs \circ \pi_1(\sigma) = \sigma$ and 
$\pi_\obs \circ \pi_1(\alpha) = \pi_\obs \circ \pi_1(\beta) \sigma$. It can be shown in the same way as in Case 2 that there exist $\beta_1' \in \Sigma_1^*$ and $\beta' \in \Sigma_\GGG^*$ satisfying \eqref{lem CSO1 eq2.5}, $\alpha_1' =  \beta_1'\sigma$, and $\pi_\obs(\beta') = \pi_\obs(\beta)$. 
By setting $\alpha'=\beta'\sigma$, we have $\pi_\obs(\alpha') = \pi_\obs(\beta')\sigma  = \pi_\obs(\beta)\sigma = \pi_\obs(\alpha) $.
Moreover, we have $\pi_{\GGG_J}(\alpha) = \pi_{\GGG_J}(\beta)$ as $\sigma \notin \Sigma_{\GGG_J}$.
By \eqref{defn product eq2}, \eqref{defn eq delta_star}, and \eqref{lem CSO1 eq2.5}, %and Lemma~\ref{lem product1}, 
% \begin{align*} 
%\begin{split}
%\Delta^*&(x_1', x_J, \alpha' = \beta' \sigma)\\
%  &=(\delta_1^*(x_1', \alpha_1' = \beta_1' \sigma),   \delta_{\GGG_J}^*(x_J,\pi_{\GGG_J}(\alpha) = \pi_{\GGG_J}(\beta)) ).
%%\Delta^*(x_1', x_J, \alpha')  
%%&= (\delta_1^*(x_1', \beta_1' \sigma),   \delta_{\GGG_J}^*(x_J,\pi_{\GGG_J}(\beta)) )\\
%%&= (\delta_1^*(x_1',  \alpha_1'),   \delta_{\GGG_J}^*(x_J,\pi_{\GGG_J}(\alpha)) )\\
%%&=(\delta_1^*(x_1', \alpha_1'), s_2, \ldots, s_n),
%\end{split}
%\end{align*} 
\begin{align*}  
\Delta^*(x_1', x_J, \alpha') &= \Delta^*(x_1', x_J,\beta' \sigma) \\ 
&= (\delta_1^*(x_1', \beta_1' \sigma),   \delta_{\GGG_J}^*(x_J,\pi_{\GGG_J}(\beta )) )\\
&= (\delta_1^*(x_1', \alpha_1' ),   \delta_{\GGG_J}^*(x_J,\pi_{\GGG_J}(\alpha)) ).  
\end{align*} 
Thereby, \eqref{lem CSO1 eq2} holds in this case.
%which implies  \eqref{lem CSO1 eq2}. 

\item Case 4: $\sigma \in \Sigma_1 \setminus (\Sigma_{\GGG_J}  \cup \Sigma_\obs)$.
Since $
\pi_\obs (\sigma) = \varepsilon$, 
$\pi_\obs \circ \pi_1(\sigma) = \varepsilon$. By \eqref{lem CSO1 eq1}, $\pi_\obs(\alpha_1') = \pi_\obs \circ \pi_1(\alpha) = \pi_\obs \circ \pi_1(\beta)$.
By the I.H., there exists $\beta'$ satisfying \eqref{lem CSO1 eq4} and 
%\begin{align}\label{lem CSO1 eq3}
$\pi_\obs(\beta') = \pi_\obs(\beta)$. 
%\end{align} 
%\begin{align} \label{lem CSO1 eq4}
%\Delta^*(x_1', x_J, \beta') 
% = (\delta_1^*(x_1', \alpha_1'),   \delta_{\GGG_J}^*(x_J,\pi_{\GGG_J}(\beta)) ).
%\end{align}  
%Since  $\sigma \in \Sigma_1 \setminus \Sigma_{\GGG_J}$, we have
%\begin{align}\label{lem CSO1 eq5}
%\delta_{\GGG_J}^*(x_J,\pi_{\GGG_J}( \beta \sigma)) =\delta_{\GGG_J}^*(x_J,\pi_{\GGG_J}(\beta)).
%%\begin{split}
%%\delta_{\GGG_J}^*(x_J,\pi_{\GGG_J}(\alpha)) &= \delta_{\GGG_J}^*(x_J,\pi_{\GGG_J}(\beta \sigma))\\
%%&=\delta_{\GGG_J}^*(x_J,\pi_{\GGG_J}(\beta)).
%%\end{split}
%\end{align}  
Since $\pi_\obs (\sigma) = \varepsilon$, 
we have 
\[\pi_\obs(\beta') = \pi_\obs(\beta')\pi_\obs(\sigma) = \pi_\obs(\beta \sigma) = \pi_\obs(\alpha).\]

By setting $\alpha' = \beta'$, we have $\pi_\obs(\alpha') =  \pi_\obs(\alpha)$.
Since  $\sigma \in \Sigma_1 \setminus \Sigma_{\GGG_J}$, $\delta_{\GGG_J}^*(x_J,\pi_{\GGG_J}( \beta \sigma)) =\delta_{\GGG_J}^*(x_J,\pi_{\GGG_J}(\beta))$.
Then, by \eqref{defn product eq2}, \eqref{lem CSO1 eq4}, and \eqref{defn eq delta_star}, %and \eqref{lem CSO1 eq5}, %, and Lemma~\ref{lem product1},
\begin{align*} 
\begin{split}
\Delta^*(x_1',x_J, \alpha')
&= \Delta^*(x_1',x_J, \beta') \\
&= (\delta_1^*(x_1', \alpha_1'),   \delta_{\GGG_J}^*(x_J,\pi_{\GGG_J}(\beta)) )\\ 
&= (\delta_1^*(x_1', \alpha_1'),   \delta_{\GGG_J}^*(x_J,\pi_{\GGG_J}(\alpha)) ).\\ 
\end{split}
\end{align*} 
Thereby, \eqref{lem CSO1 eq2} holds in this case.
\end{itemize}
As we have $\alpha'\in \Sigma_\GGG^*$ that satisfies $\pi_\obs(\alpha') =  \pi_\obs(\alpha)$ and \eqref{lem CSO1 eq2} for all cases, the induction is concluded.
\end{proof}

\subsection{Local Current-state Opacity}\label{section ISO}
This section introduces  
sufficient conditions for the CSO of the global DES, based on the local DESs.
For a set $\SSS$ of global secret states, let $\SSS[i] = \{s[i] \mid s \in \SSS \}$
be the set of its corresponding local secret states in the local DES $G_i$.
\begin{thm} \label{thm CSO2} 
We assume Assumption~\ref{asm Sig_obs_shared}.
Given a secret subset $\SSS \subseteq \XXX$ of global states, if there exists $i \in \{1, \ldots, n\}$ where $G_i$ is CSO w.r.t. 
$\SSS[i]$,
then $\GGG$ is also CSO w.r.t. $\SSS$.
\end{thm}
\begin{proof} 
We assume that $G_i$ is CSO w.r.t. $\SSS[i]$ and will show that $\GGG$ is CSO w.r.t. $\SSS$.
Let us consider any 
\begin{equation*} 
(x, \alpha) \in \XXX_\ini \times \Sigma_\GGG^* \text{ such that } \Delta^*(x, \alpha) = s \in \SSS.
\end{equation*}
By Lemma~\ref{lem product1},  
$
 \delta_i^*(x[i], \pi_i(\alpha)) = s[i] \in \SSS[i].
$
Since $G_i$ is CSO w.r.t. $\SSS[i]$, by Definition~\ref{defn CSO}, there exists $(x_i',\alpha_i') \in X_i \times \Sigma_i^*$ such that 
\[\pi_\obs(\alpha_i') = \pi_\obs \circ \pi_i(\alpha) \text{ and }
 \delta_i^*(x_i', \alpha_i') \in X_i \setminus \SSS[i].
\]
By Lemma~\ref{lem CSO1}, there exists $\alpha'$ with $\pi_\obs(\alpha) = \pi_\obs(\alpha')$ and
\begin{align*} 
\begin{split}
\Delta^*&(x[1],\ldots,x[i-1],x_i', x[i+1], \ldots, x[n], \alpha')\\
&= (s[1], \ldots, s[i-1], \delta_i^*(x_i', \alpha_i'), s[i+1], \ldots, s[n])\\
&\in \XXX \setminus \SSS.
\end{split}
\end{align*}
Thereby, $\GGG$ is CSO w.r.t. $\SSS$ and the theorem holds.
\end{proof}

\begin{figure}[t]
  \centering
    \includegraphics[width=0.8\linewidth]{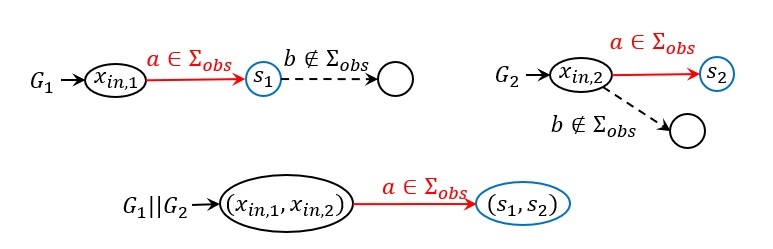}
  \caption{Two local DESs and the accessible part of their parallel composition. 
  The DES $G_1$ is CSO w.r.t $\{s_1\}$, but $G_1\parallel G_2$ is not CSO w.r.t. $\{(s_1, s_2)\}$.}
  \label{fig: remark CSO2} 
\end{figure}

\begin{remark}\label{remark CSO Asm1}
Assumption~\ref{asm Sig_obs_shared} is a necessary condition for Theorem \ref{thm CSO2}.
In the example DES in Fig.~\ref{fig: remark CSO2}, in which the shared event $b$ is not observable,
 $G_1$ is CSO w.r.t. $\{s_1\}$  
but $G_1 \parallel G_2$ is not CSO w.r.t. $\{s_1, s_2\}$.
In this example, the shared event $b$ is blocked by the parallel composition. 
Such an event blocking is generally difficult to detect without constructing any part of $\GGG$, which we aim to avoid. 
%It is difficult to verify whether or not there exists such blocked events without constructing the global DES or the observers of local DESs.
\end{remark}

\begin{figure}[t]
  \centering
    \includegraphics[width=0.6\linewidth]{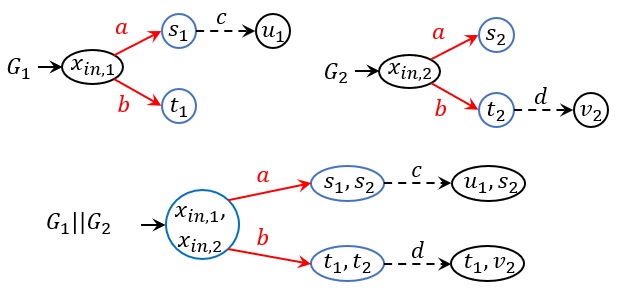}
  \caption{Two local DESs and the accessible part of their parallel composition. 
	$\Sigma_\obs = \{a,b\}$.  
 The DESs $G_1$ and $G_2$ are not CSO w.r.t. $\{s_1, t_1\}$ and $\{s_2, t_2\}$, respectively, but their composition $G_1 \parallel G_2$ is CSO w.r.t. $\{(s_1, s_2), (s_1, t_2), (t_1, s_2), (t_1, t_2)\}$.}
  \label{fig: remark CSO4} 
\end{figure}

\begin{figure}[t]
  \centering
    \includegraphics[width=0.7\linewidth]{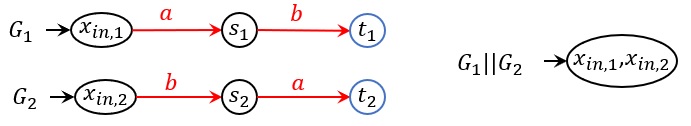}
  \caption{Two local DESs and the accessible part of their parallel composition. All events are observable. The DESs  $G_1$ and $G_2$ are not CSO w.r.t. $\{t_1\}$ and $\{t_2\}$, respectively, but  $G_1 \parallel G_2$ is CSO w.r.t. $\{(t_1, t_2)\}$.}
  \label{fig: remark CSO5}
\end{figure}

\begin{remark}\label{remark CSO3}
The inverse of the implication in Theorem~$\ref{thm CSO2}$ does not hold.
In other words,
the global DES $\GGG$ being  CSO w.r.t. $\SSS$ does not imply the existence of a local DES $G_i$ that is CSO w.r.t. $\SSS[i]= \{s[i] \mid s \in \SSS \}$. 
We provide two counter examples in Fig.~\ref{fig: remark CSO4} and Fig.~\ref{fig: remark CSO5}. 
In both examples, each local DES $G_i$, $i \in \{1,2\}$, is not CSO w.r.t. $\SSS[i]$, but the global DES $\GGG$ is CSO w.r.t. $\SSS$. 
From both examples, we can see that the inverse of the implication in Theorem~$\ref{thm CSO2}$ does not hold event if $\SSS = \SSS[i] \times \ldots \times \SSS[n]$. 
Notice that, in Fig.~\ref{fig: remark CSO4}, the global DES becomes CSO thanks to unobservable local events $c$ and $d$.
In Fig.~\ref{fig: remark CSO5}, the event sequence reaching the secret state $(t_1, t_2)$ is blocked by the parallel composition. 
%As discussed in Remark \ref{remark CSO Asm1}, detecting such a blocked event sequence is difficult without constructing any part of $\GGG$.
\end{remark}

Theorem~\ref{thm CSO2} provides a sufficient condition for the opacity of the global DES.
However, if there is no local DES $G_i$ that is CSO w.r.t. $\SSS[i]$,
 $\GGG$ can still be is CSO w.r.t. $\SSS$. 
Using Lemma~\ref{lem S1S1}, 
we propose another sufficient  condition for the CSO of
$\GGG$ by considering subsets of secret states.
%its  secret subsets. %$\SSS_1, \ldots, \SSS_m \in \SSS$ where $\SSS = \SSS_1 \cup \ldots \cup \SSS_m$.

\begin{thm} \label{thm CSO3} 
We assume Assumption~\ref{asm Sig_obs_shared}.
Consider a  set of secret global states $\SSS = \SSS_1 \cup \ldots \cup \SSS_m \subseteq \XXX$. 
Let $\SSS[i] = \{s[i] \mid s \in \SSS\}$ and $\SSS_j[i] = \{s[i] \mid s \in \SSS_j\}$ for all $j \in \{1,\ldots,m\}$.
Then, $\GGG$ is CSO w.r.t. $\SSS$ if 
for all $ {j \in \{1,\ldots, m\}}$,
there exists $ {i \in \{1,\ldots,n\}}$ such that the following condition holds.
%Suppose that,
% for all global secret subset $\SSS_{j \in \{1,\ldots, m\}}$,
%there exists a local DES $G_{i \in \{1,\ldots,n\}}$ such that
\begin{align}\label{thm CSO3 eq1} 
\begin{split}
\forall (x_i, \alpha_i) \in X_{i,\ini} \times \Sigma_i^*, \delta_i^*&(x_i, \alpha_i) \in \SSS_j[i],\\
\exists (x_i', \alpha_i') \in X_{i,\ini} \times \Sigma_i^*, &\delta_i^*(x_i', \alpha_i') \in X_i \setminus \SSS[i]\\
&\text{ and }\pi_\obs(\alpha_i') = \pi_\obs(\alpha_i).
\end{split}
\end{align}
\end{thm}
\begin{proof} 
Let us consider a secret subset  $\SSS_{j}$ and 
let $G_{i}$ be the local DES that satisfy \eqref{thm CSO3 eq1}.
Consider any
\begin{equation} \label{thm CSO3 eq2} 
(x, \alpha) \in \XXX_\ini \times \Sigma_\GGG^* \text{ such that } \Delta^*(x, \alpha) = s \in \SSS_j.
\end{equation}
By Lemma~\ref{lem product1},  
$
 \delta_i^*(x[i], \pi_i(\alpha)) = s[i] \in \SSS_j[i].
$ 
By \eqref{thm CSO3 eq1}, there exists $(x_i',\alpha_i') \in X_{i,\ini} \times \Sigma_i^*$ such that 
\[\pi_\obs(\alpha_i') = \pi_\obs \circ \pi_i(\alpha) \text{ and }
 \delta_i^*(x_i', \alpha_i') \in X_i \setminus \SSS[i].
\] 
By Lemma~\ref{lem CSO1}, there exists $\alpha'$ with $\pi_\obs(\alpha) = \pi_\obs(\alpha')$ and
\begin{align}   
\Delta^*&(x[1],\ldots,x[i-1], \delta_i^*(x_i', \alpha_i')   , x[i+1], \ldots, x[n], \alpha')\nonumber\\ 
&\in \XXX \setminus \SSS.\label{thm CSO3 eq3} 
\end{align}
By Lemma~\ref{lem S1S1}, \eqref{thm CSO3 eq2}  and \eqref{thm CSO3 eq3},
the global DES
$\GGG$ is CSO w.r.t. $\SSS = \bigcup_{j \in \{1,\ldots,m\}} S_j$ and the theorem holds.
\end{proof}

Theorem~\ref{thm CSO3} also provides a sufficient condition. Its inverse of the implication does not hold,
as shown in the counter example in Fig. \ref{fig: remark CSO5}.
However, we show in Section~\ref{section algo} that we can use this theorem to verify the global DES in some cases.

 %in the same way as in Remark~\ref{remark CSO3}.
%However, for the example in Fig. \ref{fig: remark CSO4}, we may consider
%$\SSS_1 = \{(s_1, s_2), (s_1, t_2)\}$ and $\SSS_2 = \{(t_1, s_2), (t_1, t_2)\}$.

\subsection{Local Initial-state Opacity}
We show  that our results for CSO also hold for ISO.

\begin{thm}\label{thm ISO1} 
We assume Assumption \ref{asm Sig_obs_shared}.
Given a secret subset $\SSS \subseteq \XXX_\ini$ of global initial states, if there exists $i \in \{1,\ldots,n\}$ such that $G_i$ is ISO w.r.t. $\SSS[i] = \{s[i] \mid s \in \SSS \}$, then $\GGG$ is also ISO w.r.t. $\SSS$.
\end{thm}
\begin{proof} 
Suppose that $G_i$ is ISO w.r.t. $\SSS[i]$.
Consider any %pair
\begin{equation}\label{thm ISO1 eq1} 
(x, \alpha) \in \SSS \times \Sigma_\GGG^* \text{ such that } \Delta^*(x, \alpha)!.
\end{equation} 
By Lemma \ref{lem product1}, $\delta^*(x[i], \pi_i(\alpha))!$. 
Since $x[i] \in \SSS[i]$ and $G_i$ is ISO w.r.t. $\SSS[i]$, there exists $(x_i', \alpha_i') \in (X_i \setminus \SSS[i]) \times \Sigma_i^*$ with
\begin{equation*} 
\delta^*(x_i', \alpha_i')! \text{ and } \pi_\obs(\alpha_i')= \pi_\obs\circ \pi_i(\alpha).
\end{equation*} 
Let $x' = x[1],\ldots,x[i-1],x_i',x[i+1],\ldots,x[n]$.
By Lemma \ref{lem CSO1}, there exists $\alpha'$ such that
\begin{equation}\label{thm ISO1 eq2} 
\Delta^*(x', \alpha')! \text{ and } \pi_\obs(\alpha) = \pi_\obs(\alpha').
\end{equation} 
Notice that $x' \in \XXX \setminus \SSS$ because $x_i' \in X_i \setminus \SSS[i]$.
Therefore, the lemma holds by \eqref{thm ISO1 eq1}, \eqref{thm ISO1 eq2}, and Definition~\ref{defn ISO}. 
\end{proof}
\begin{figure}[t]
  \centering
    \includegraphics[width=0.7\linewidth]{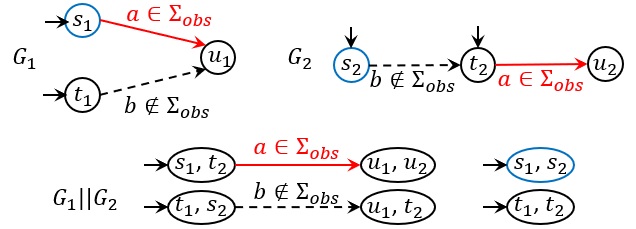}
  \caption{Two local DESs and the accessible part of their parallel composition. Initial states of $G_1$ (\emph{resp.} $G_2$) are $s_1$ and $t_1$ (\emph{resp.} $s_2$ and $t_2$).
   The DES $G_2$ is ISO w.r.t $\{s_2\}$, but $G_1\parallel G_2$ is not ISO w.r.t. $\{(s_1, s_2)\}$.}
  \label{fig: remark ISO3}
\end{figure}

%\begin{remark}\label{remark ISO generalize}
%Theorem~\ref{thm ISO1} is a generalized version of Theorem 4 in \cite{saboori2010reduced}, which restricts $\SSS= \SSS[i] \times \cdots \times \SSS[i]$. 
%\end{remark}

\begin{remark}\label{remark ISO Asm1}
Assumption~\ref{asm Sig_obs_shared} is a necessary condition for Theorem \ref{thm ISO1}.
In the example in Fig.~\ref{fig: remark ISO3}, $G_2$ is ISO w.r.t $\{s_2\}$, but $G_1\parallel G_2$ is not ISO w.r.t. $\{(s_1, s_2)\}$. %Notice at the outgoing events from states $(s_1, t_1)$ and $(s_2, t_2)$ are blocked by the parallel composition. 
\end{remark}

\begin{remark}\label{remark ISOnotcomplete}
The inverse of the implication in Theorem~\ref{thm ISO1} also does not hold.
The global DES $\GGG$ being  ISO w.r.t. $\SSS$ does not imply the existence of a local DES $G_i$ that is
 ISO w.r.t. $\SSS[i]$.
 Fig.~\ref{fig: remark ISO6} provides a counter example. Both local DESs $G_1$ and $G_2$ are not ISO w.r.t $\{s_1, t_1\}$ and $\{s_2, t_2\}$, respectively, but $G_1\parallel G_2$ is ISO w.r.t. $\{(s_1, s_2), (t_1, t_2)\}$. 
\end{remark}

\begin{figure}[t]
  \centering
    \includegraphics[width=0.6\linewidth]{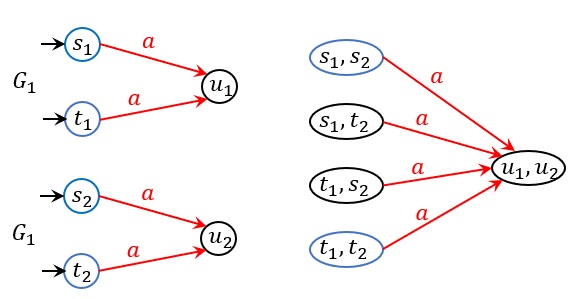}
  \caption{Two local DESs and the accessible part of their parallel composition. Initial states of $G_1$ (\emph{resp.} $G_2$) are $s_1$ and $t_1$ (\emph{resp.} $s_2$ and $t_2$).
  The event $a$ is observable.
   The DESs $G_1$ and $G_2$ are both not ISO w.r.t $\{s_1, t_1\}$ and $\{s_2, t_2\}$, respectively, but $G_1\parallel G_2$ is ISO w.r.t. $\{(s_1, s_2), (t_1, t_2)\}$.}
  \label{fig: remark ISO6}
\end{figure}
 
In the same way as in Theorem~\ref{thm CSO3}, we propose another sufficient condition in  for the ISO of
$\GGG$ by considering
its  secret subsets $\SSS_1, \ldots, \SSS_m \in \SSS$ where $\SSS = \SSS_1 \cup \ldots \cup \SSS_m$. 

\begin{thm} \label{thm ISO3} 
We assume Assumption~\ref{asm Sig_obs_shared}.
Consider a  set of secret global initial states $\SSS = \SSS_1 \cup \ldots \cup \SSS_m \subseteq \XXX_\ini$. 
Let $\SSS[i] = \{s[i] \mid s \in \SSS\}$ and $\SSS_j[i] = \{s[i] \mid s \in \SSS_j\}$ for all $j \in \{1,\ldots,m\}$.
Then, $\GGG$ is ISO w.r.t. $\SSS$ if
 for all $j \in \{1,\ldots, m\}$,
there exist $ i \in \{1,\ldots,n\}$ satisfying the following condition.
\begin{align}\label{thm ISO3 eq1} 
\begin{split}
\forall (x_i, \alpha_i) \in \SSS_j[i] \times &\Sigma_i^*, \delta_i(x_i,  \alpha_i)!,\\
\exists (x_i', \alpha_i') \in (X_i \setminus &\SSS[i]) \times \Sigma_i^*, \delta_i(x_i', \alpha_i')!\\
&\text{ and }\pi_\obs(\alpha_i') = \pi_\obs(\alpha_i).
\end{split}
\end{align}

\end{thm}
\begin{proof} 
Consider any secret subset  $\SSS_{j}$ and 
let $G_{i}$ be the local DES that satisfy \eqref{thm ISO3 eq1}.
Consider any
\begin{equation} \label{thm ISO3 eq2} 
(x, \alpha) \in \SSS_j \times \Sigma_\GGG^* \text{ such that } \Delta^*(x, \alpha)!
\end{equation}
By Lemma~\ref{lem product1},  
$
 \delta_i^*(x[i], \pi_i(\alpha))!$ and $x[i] \in \SSS_j[i].
$ 
By \eqref{thm ISO3 eq1}, there exists $(x_i',\alpha_i') \in (X_i \setminus \SSS[i]) \times \Sigma_i^*$ such that 
\[\pi_\obs(\alpha_i') = \pi_\obs \circ \pi_i(\alpha) \text{ and }
 \delta_i^*(x_i', \alpha_i')!
\] 
Let $x' = x[1],\ldots,x[i-1],x_i',x[i+1],\ldots,x[n]$.
By Lemma \ref{lem CSO1}, there exists $\alpha'$ such that
\begin{equation}\label{thm ISO3 eq3} 
\Delta^*(x', \alpha')! \text{ and } \pi_\obs(\alpha) = \pi_\obs(\alpha').
\end{equation} 
Notice that $x' \in \XXX \setminus \SSS$ because $x_i' \in X_i \setminus \SSS[i]$.  
By Lemma~\ref{lem S1S2}, \eqref{thm ISO3 eq2}  and \eqref{thm ISO3 eq3}, 
%the global DES
$\GGG$ is ISO w.r.t. $\SSS = \bigcup_{j \in \{1,\ldots,m\}} S_j$. %and the theorem holds.
\end{proof}

The inverse of the implication in Theorem~\ref{thm ISO1} also does not hold, as 
it can be shown using the counter example in Fig.~\ref{fig: remark ISO6}. 

\section{Opacity Verification of Global System} \label{section algo}
In section \ref{section local}, we presented the sufficient conditions of the opacity (CSO and ISO) of the global system $\GGG$, by only considering the local DESs $G_i$, $i \in \{1,\ldots,n\}$.
Hence, we can use Theorems~\ref{thm CSO2} and \ref{thm ISO1} is to verify each local DES using existing opacity verification algorithms (e.g. \cite{tong2018current, dubreil2010supervisory, saboori2011opacity, yin2015uniform}). 
Then, if there exists a local DES $G_i$ that is opaque w.r.t. $\SSS[i]$, the global DES $\GGG$ is also opaque w.r.t. $\SSS$ thanks to Theorems~\ref{thm CSO2} and \ref{thm ISO1}. 
By using this technique, we only need to construct the intruder's observer automata \cite{cassandras2009} for each local DES $G_i$, not the global DES $\GGG$.
As a result, we can reduce the size of the intruder's observer automata from %complexity of opacity verification from 
$\mathcal{O}(2^{|X_1|\times \ldots \times |X_n|})$ (for $\GGG$) to
 $\mathcal{O}(2^{|X_1|} + \ldots + 2^{|X_n|})$. %, assuming that event sets $\Sigma_i$ are small.

For example, let us consider the global DES $\mathit{Agent1}\parallel \mathit{Agent2}\parallel \mathit{Resource}$ of the DESs in Fig.~\ref{fig: example}. 
Suppose that $(\mathit{1wait}, \mathit{2use}, \mathit{2})$ is the only secret state.
%For this case, $\mathit{Agent1}$ is CSO w.r.t. $\{\mathit{1wait}\}$. 
By considering the observer automaton in Fig.~\ref{fig: observer} (a), 
we know that $\mathit{Agent1}$ is CSO (\emph{resp.} ISO) w.r.t. $\{\mathit{1wait}\}$.
By Theorem~\ref{thm CSO2}, the global DES is also CSO (\emph{resp.} ISO)
w.r.t. $\{(\mathit{1wait}, \mathit{2use}, \mathit{2})\}$.

\begin{figure}[t]
  \centering
    \includegraphics[width=0.7\linewidth]{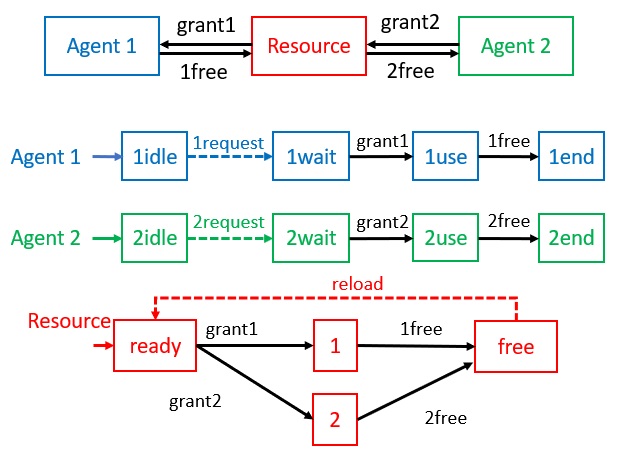}
  \caption{Two agents sharing one resource. The events ``1request'', ``2request'', and ``reload'' are not observable by the intruder.}
  \label{fig: example}
\end{figure}

 As discussed in Remarks \ref{remark CSO3} and \ref{remark ISOnotcomplete}, Theorems~\ref{thm CSO2} and \ref{thm ISO1} provide only sufficient conditions for the opacity of the global DES $\GGG$. 
%If there is no local DES $G_i$ that is opaque w.r.t. $\SSS[i]$,
% $\GGG$ can still be opaque w.r.t. $\SSS$. 
Let us consider the DESs in Fig.~\ref{fig: example} and $\SSS = \{ (\mathit{1wait},\mathit{2use}, \mathit{2}),  (\mathit{1use},\mathit{2wait}, \mathit{1})
(\mathit{1end}, \mathit{2end}, \mathit{free}) \}$.
For this case, we cannot verify the global DES by verifying local DESs w.r.t. theirs corresponding local secret sets, e.g., $\mathit{Agent1}$ is not opaque w.r.t. $\{\mathit{1wait}, \mathit{1use}, \mathit{1end}\}$.
However, by Theorems~\ref{thm CSO3} and \ref{thm ISO3}, we can try to verify the opacity of $\GGG$ by verifying the opacity of each $G_i$ w.r.t. $\{s\}$, for all $s \in \SSS$.
Let $\SSS_1 = \{ (\mathit{1wait},\mathit{2use}, \mathit{2}) \}$,
$\SSS_2=\{(\mathit{1use},\mathit{2wait}, \mathit{1})\}$,
and  $\SSS_3 = \{ (\mathit{1end}, \mathit{2end}, \mathit{free}) \}$.
Let $G_1$, $G_2$, and $G_3$ be  $\mathit{Agent1}$, $\mathit{Agent2}$, and $\mathit{Resource}$, respectively.
Using the observer automata in Fig.~\ref{fig: observer}, we have the following properties.
\begin{enumerate}
\item $\SSS_1[1] = \{ \mathit{1wait}\}$. For the event sequence $\mathit{1request}$ with $\delta_1(\mathit{1idle}, \mathit{1request}) = \mathit{1wait} \in \SSS_1[1]$, we have $\delta_1(\mathit{1idle}, \varepsilon) = \mathit{1idle} \notin \SSS[1] = \{\mathit{1wait}, \mathit{1use}, \mathit{1end}\}$ and $\pi_\obs(1request) = \pi_\obs(\varepsilon) = \varepsilon$.
\item $\SSS_2[2] = \{ \mathit{2wait}\}$. For the event sequence $\mathit{2request}$ with $\delta_2(\mathit{2idle}, \mathit{2request}) = \mathit{2wait} \in \SSS_2[2]$, we have $\delta_2(\mathit{2idle}, \varepsilon) = \mathit{2idle} \notin \SSS[2] = \{\mathit{2wait}, \mathit{2use}, \mathit{2end}\}$ and $\pi_\obs(2request) = \pi_\obs(\varepsilon) = \varepsilon$.
\item $\SSS_3[3] = \{ \mathit{free}\}$. For all $\alpha \in \Sigma_3^*$
%$\in \{ \mathit{grant1\, free1}, \mathit{grant2\, free2}\}$ 
such that $\delta_3(\mathit{ready}, \alpha) = \mathit{free} \in \SSS_3[3]$, 
we have $\alpha' = \alpha\,\mathit{reload}$ where 
$\delta_3(\mathit{ready}, \alpha') = \mathit{ready} \notin \SSS[3]=\{1,2,\mathit{free}\}$ and
$\pi_\obs(\alpha) = \pi_\obs(\alpha')$. 
\end{enumerate}
By Theorem~\ref{thm CSO3}, the global DES is CSO w.r.t $\SSS$.
We can use Theorem~\ref{thm ISO3} to verify that the  global DES is also ISO w.r.t $\SSS$ in the same way.
Thus, for this example, we can verify the global DES using the observer automata of the local DESs.

\begin{figure}[t]
  \centering
    \includegraphics[width=0.8\linewidth]{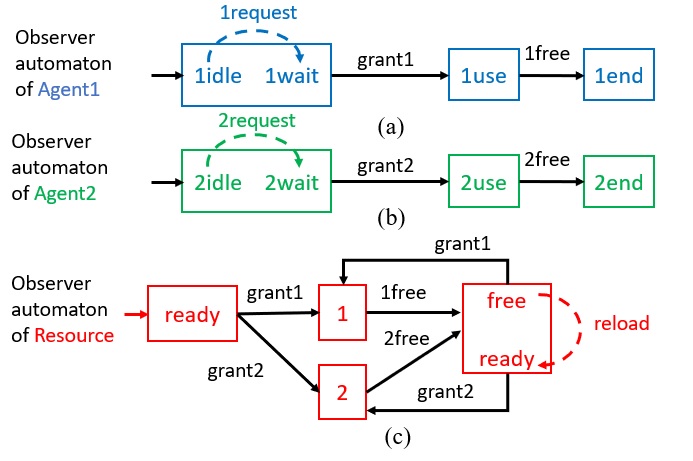}
  \caption{Intruder's CSO observer automata for local DESs in Fig.~\ref{fig: example}. %(a) The observer for $\mathit{Agent1}$.
  %(b) The observer of $\mathit{Agent2}$. (c) The observer of $\mathit{Resource} $.
  \label{fig: observer}}
\end{figure}

\section{Conclusions} \label{section conclusion}

We studied current-state opacity (CSO) and initial-state opacity (ISO) verification of a distributed DES. 
The distributed DES, which we call the global DES, is the parallel composition of $n$ local DESs. %, each of which represents a sub-component of the global DES.
By assuming that the intruder observes the events shared between local DESs,
we proposed sufficient conditions for the opacity (CSO and ISO) of the global DES, by considering only the opacity of local DESs. Using these conditions, we introduced efficient methodologies
to verify the global DES's opacity 
using observer automata of the local DESs.
%without constructing the intruder's observer automaton of the global DES.
For future work, we will study the verification of other system opacity concepts and the opacity enforcement of distributed DESs. 
%It is also a future work to the compositional verification and enforcement approaches. 
%\cite{mohajerani2019compositional, noori2018compositional, mohajerani2019transforming}
%We will also study...

%The common technique for modular systems opacity verification is to construct data structures to estimate the intruder's information of the global DES based on observed event sequences, e.g., the intruder's estimate function in \cite{tong2019current} and the Augmented I-Observer in \cite{tong2018decentralized}.
%These data structures are constructed using the global DES, and can be very large and costly. 
%In Sections ~\ref{section local} and \ref{section algo}, we propose 
%sufficient conditions and corresponding efficient algorithms
% for the opacity verification  problem that avoid constructing the global DES.

\addtolength{\textheight}{-12cm}   % This command serves to balance the column lengths

%%%%%%%%%%%%%%%%%%%%%%%%%%%%%%%%%%%%%%%%%%%%%%%%%%%%%%%%%%%%%%%%%%%%%%%%%%%%%%%%

%%%%%%%%%%%%%%%%%%%%%%%%%%%%%%%%%%%%%%%%%%%%%%%%%%%%%%%%%%%%%%%%%%%%%%%%%%%%%%%%

%%%%%%%%%%%%%%%%%%%%%%%%%%%%%%%%%%%%%%%%%%%%%%%%%%%%%%%%%%%%%%%%%%%%%%%%%%%%%%%%
%\section*{APPENDIX}

%\section*{ACKNOWLEDGMENT}
%
%The preferred spelling of the word ÒacknowledgmentÓ in America is without an ÒeÓ after the ÒgÓ. Avoid the stilted expression, ÒOne of us (R. B. G.) thanks . . .Ó  Instead, try ÒR. B. G. thanksÓ. Put sponsor acknowledgments in the unnumbered footnote on the first page.
%
%
%
%%%%%%%%%%%%%%%%%%%%%%%%%%%%%%%%%%%%%%%%%%%%%%%%%%%%%%%%%%%%%%%%%%%%%%%%%%%%%%%%%
%
%References are important to the reader; therefore, each citation must be complete and correct. If at all possible, references should be commonly available publications.

%\begin{thebibliography}{99}
%
%\bibitem{textbook1}
%  Any textbook about automata, e.g., http://people.rennes.inria.fr/Eric.Fabre/Module\%20MAD/chapter2\%20-\%20CL\%20.pdf, or ``Introduction to Discrete Event Systems''.
%
%
%
%
%
%
%\end{thebibliography}
%\bibliographystyle{IEEEtran}
%\bibliography{refs}

\end{document}